\documentclass[11pt]{amsart}
\usepackage{amsthm,amsmath,amssymb,url}
\usepackage[top=15mm, bottom=20mm, left=30mm, right=30mm]{geometry}
\usepackage{xcolor}
\usepackage{tikz}
\usepackage{fourier}

\usepackage{version}
\usepackage{enumerate}

\numberwithin{figure}{section}

\newcommand{\nats}{\mathbb{N}}

\newcommand{\rats}{\mathbb{Q}}

\newcommand{\reals}{\mathbb{R}}
\newcommand{\preals}{\ensuremath{\reals_+}}

\DeclareMathOperator{\sunam}{\mathbf{S}}
\DeclareMathOperator{\wunam}{\mathbf{W}}
  
\DeclareMathOperator{\cond}{\mathbf{C}}

\newcommand{\cons}{\mathcal{K}}
\newcommand{\profs}{\mathcal{P}}
\newcommand{\votesits}{\mathcal{V}}
\newcommand{\votedists}{\mathcal{F}}

\newcommand{\elecs}{\mathcal{E}}

\DeclareMathOperator{\bis}{\beta}
\DeclareMathOperator{\simp}{\Delta}

\DeclareMathOperator{\R}{\mathcal{R}}
\DeclareMathOperator{\rk}{rk}

\DeclareMathOperator{\nummap}{\mathcal{N}}
\DeclareMathOperator{\distmap}{\mathcal{D}}

\newcommand{\hl}[1]{\textnormal{\textbf{#1}}} 
\newcommand{\commentred}[1]{\textcolor{red}{\textit{#1}}}

\theoremstyle{plain}
\newtheorem{defn}{Definition}[section]
\newtheorem{remark}[defn]{Remark}

\newtheorem{prop}[defn]{Proposition}

\newtheorem{cor}[defn]{Corollary}
\newtheorem{eg}[defn]{Example}

\begin{document}

\title{Distance rationalization of anonymous and homogeneous voting rules}
\author{Benjamin Hadjibeyli}
\address{ENS Lyon}
\email{benjamin.hadjibeyli@ens-lyon.org}
\author{Mark C. Wilson}
\address{University of Auckland}
\email{mcw@cs.auckland.ac.nz}

\begin{abstract}
The concept of distance rationalizability of voting rules has been explored in recent years by several authors. Roughly speaking, we first choose a \emph{consensus set} of \emph{elections} (defined via preferences of voters over candidates) for which the result is specified \emph{a priori} (intuitively, these are elections on which all voters can easily agree on the result). We also choose a measure of \emph{distance} between elections. The result of an election outside the consensus set is defined to be the result of the closest consensual election under the distance measure. 

Most previous work has dealt with a definition in terms of preference profiles. However, most voting rules in common use are anonymous and homogeneous. In this
case there is a much more succinct representation (using the voting simplex) of the inputs to the rule. This representation has been widely used in the voting literature, but rarely in the context of distance rationalizability. 

We show exactly  how to connect  distance rationalizability on profiles for anonymous and homogeneous rules to geometry in the simplex. We develop the connection for the important special case of votewise distances, recently introduced and studied by Elkind, Faliszewski and Slinko in several papers. This yields a direct interpretation in terms of well-developed mathematical concepts not seen before in the voting literature, namely Kantorovich (also called Wasserstein) distances and the geometry of 
Minkowski spaces. 

As an application of this approach, we prove some positive and some negative results about the decisiveness of distance rationalizable anonymous and homogeneous rules. The positive results connect with the recent theory of hyperplane rules, while the negative ones deal with distances that are not metrics, controversial notions of consensus, and the fact that the $\ell^1$-norm is not strictly convex.

We expect that the above novel geometric interpretation will aid the analysis of rules defined by votewise distances, and the discovery of new rules with desirable properties.
\end{abstract}


\thanks{MW thanks Elchanan Mossel and Miklos Racz for research hospitality and useful
conversations.}
\maketitle

\section{Introduction} \label{s:intro}

We are interested in the relation between two ways of describing voting rules (interpreted broadly), each of which has a geometric flavour.

The class of anonymous and homogeneous voting rules includes all rules used in practice, and most rules appearing in the research literature (Dodgson's rule is a notable exception). For such rules there is an obvious concise way to describe an input profile of preferences, using the \emph{vote simplex}. This approach goes back at least as far as Young \cite{Youn1975} and was extensively developed and popularized by Saari
\cite{Saar1994}. By allowing us to use geometric intuition, it aids in the analysis of many properties of anonymous and homogeneous voting rules.

The framework of \emph{distance rationalizability} is a useful way to organize the huge number of voting rules that have been introduced. By decomposing a rule into a \emph{consensus} and a notion of \emph{distance} to that consensus, the framework allows systematic derivation of axiomatic properties of the rule from those of its components. 
This kind of analysis has been  carried out recently by Elkind, Faliszewski and Slinko \cite{EFS2012, EFS2015} and the present authors \cite{HaWi2016}, 
following early work by Campbell, Lerer and Nitzan \cite{Nitz1981, LeNi1985, CaNi1986} and Meskanen and Nurmi \cite{MeNu2008}.

Until now, the two approaches described above have not been explicitly connected. Specific distance-based rules have indeed been studied in the simplex or permutahedron, notably by Zwicker \cite{Zwic2008, Zwic2008b, CDGL+2012}. However, a 
more general approach is lacking. As shown in \cite{HaWi2016}, the theory can be developed simultaneously for social choice rules and social welfare rules, and for very general distances and consensus notions, in a way that clarifies the relationship between the profile-based and simplex-based representations.

\subsection{Outline of paper and our contribution} 
\label{ss:contrib} 

In Section~\ref{s:defs} we cover the basic notation and terminology. 
A rule defined directly on the simplex is  automatically anonymous and homogeneous. Conversely, every anonymous and homogeneous rule  can instead be defined on the simplex. In Section~\ref{s:simplex} we show how the usual distance rationalizability approach on profiles connects with the geometric approach on the vote simplex. 

The methods of Zwicker and Saari, starting with the simplex or related geometric representations of the space of preference profiles and using usual Euclidean geometry, have led both to the discovery of some new anonymous and homogeneous rules and to improved analysis of some old ones. We give an example in Proposition~\ref{prop:$ell^p$} of some new rules defined in an intuitive way. However, as it seems that in order to find more interesting rules we must dig deeper and use less well-known consensuses or distances, we focus instead on the analysis of rules defined on profile space, using our geometric machinery.

Starting with an anonymous and homogeneous rule defined via distance rationalization on profile space, we consider the corresponding rule on the simplex and interpret it geometrically.
Abstractly, this is straightforward if we use a \emph{quotient distance}, a general construction which is nontrivial to compute in general. However, in the situation of the present paper, we can give simple explicit formulae (Proposition~\ref{prop:homog simple}). We focus in Section~\ref{ss:wasserstein} on the special case of $\ell^p$-votewise
distances, which have been shown by Elkind, Faliszweski and Slinko (and the present authors) 
to have many desirable properties. We sharpen further the description above to show that in this case, the quotient distance on the simplex is a \emph{Wasserstein} (also called \emph{Kantorovich}) distance, a concept widely used in probability theory and its applications in computer science. In particular when $p=1$ (the most natural case for voting) each such distance is induced by a norm, and we can interpret everything in terms of the geometry of finite-dimensional normed spaces (also called \emph{Minkowski spaces}). This provides a new perspective to the voting literature and suggests not only new voting rules defined using geometric intuition, but also a new geometric tool for the analysis of existing rules.

In particular, in Section~\ref{s:ties} we use the simplex representation to explore the decisiveness (how often it gives a unique winner) of a distance rationalizable rule. On the positive side, we give a sufficient condition in
Corollary~\ref{cor:hyperplane} for a rule  to be a
\emph{hyperplane rule} and thus admit a vanishingly small fraction of profiles where ties occur. For example, 
any rule defined using an $\ell^p$-votewise distance and the strong or weak
unanimity consensus satisfies this condition. On the negative side, we
show in Proposition~\ref{pr:norm} that ties can occur in a large fraction of profile space if we use $\ell^1$-votewise metrics, unless the notion of consensus is very well chosen. This sheds light on some common consensus notions and casts some doubt on that of Condorcet. In Section~\ref{s:future} we make some recommendations for desirable properties of consensus sets and distances.

The approach adopted here and in \cite{HaWi2016}, following \cite{EFS2015}, allows for systematic exploration of the space of aggregation rules and the construction of rules with guaranteed axiomatic properties. We expect that more insight into distance-based voting rules will be obtained by exploiting the deeper geometric connections developed here.

\section{Basic definitions} \label{s:defs}

We use standard concepts of social choice theory. Not all of these concepts have completely 
standardized names. 

\begin{defn}
\label{def:rankings}
We fix a finite set $C = \{c_1, c_2, \dots, c_m\}$ of \hl{candidates} and an infinite set $V^*= \{v_1, v_2, \dots, \}$ of potential  
\hl{voters}. For each $s$ with $1\leq s \leq m$, an
\hl{$s$-ranking} is a strict linear order of $s$ elements chosen from
$C$. The set of all $s$-rankings is denoted $L_s(C)$. When $s=m$ we
write simply $L(C)$. When $s=1$ we identify $L_1(C)$ with $C$ in the natural way. 
\end{defn}

\begin{defn}
\label{def:profiles}
A \hl{profile} is a function $\pi: V \to L(C)$ where $V\subset V^*$ is finite. 
We denote the set of all profiles with fixed $C$ and $V$ by $\profs(C,V)$ and the set of all profiles by $\profs(C)$. 
An \hl{election} is a triple $(C,V, \pi)$
with $\pi\in \profs(C,V)$. We denote the set of all elections with fixed $C$ and $V$  by
$\elecs(C, V)$, and the class of all elections by $\elecs$. 
\end{defn}

\begin{defn}
\label{def:rules}
A \hl{social rule of size $s$} is a function $R$ that takes each election $E = (C,V, \pi)$ to
a nonempty subset of $L_s(C)$. When there is a unique $s$-ranking chosen, the
word ``rule" becomes ``function". When $s=1$, we have the usual
\hl{social choice function}, and when $s=|C|$ the usual \hl{social welfare
function}.

For each subset $D$ of $\elecs$ we can consider a \hl{partial social rule with domain $D$} 
to be defined as above, but with domain restricted to $D$.
\end{defn}

\subsection{Consensus}
\label{ss:consensus}

Intuitively, a consensus is simply a socially agreed unique outcome on some set of
elections. We now define it formally.
\begin{defn}
\label{def:cons}
An $s$-\hl{consensus} is a  partial social function $\cons$ of
size $s$.  The domain $D(\cons)$ of 
$\cons$ is called an \hl{$s$-consensus set}
and is partitioned into the inverse images $\cons_r:=
\cons^{-1}(\{r\})$.

\end{defn}

Several specific consensuses have been described in the literature. We list a few important ones. 
Some have been discussed by previous authors only in the case $s=1$ but the definitions extend naturally. 
\begin{defn} 
\label{def:cons_eg} 
We use the following consensuses in this article.
\begin{itemize}
\item  We denote  by $\sunam^s$ the consensus $\cons$ for which $\cons_r$ is the election
in which all voters agree that $r$ is the ranking of the top $s$ candidates. 
When $s=|C|$, we simply write $\sunam$ (called the 
\hl{strong unanimity consensus}),
whereas when $s = 1$, for consistency with previous authors 
we denote it $\wunam$, the \hl{weak unanimity consensus}. 
\item The $1$-\hl{Condorcet} consensus $\cond$ has domain consisting of all elections for which a Condorcet winner exists. That is, there is a candidate $c$ (the Condorcet winner) such that for every other candidate $b$
a fraction strictly greater than $1/2$ of voters rank $c$ above $b$.

\end{itemize}
\end{defn}

\subsection{Distances}
\label{ss:metrics}

We require a notion of distance on elections. We aim to be as general as possible. 

\begin{defn}(distance)
\label{def:dist}
A \hl{distance} (or \hl{hemimetric}) on $\elecs$ is a function
$d:\elecs \times \elecs \to \preals \cup \{\infty\}$ that satisfies for all $x,y,z\in \elecs$
\begin{itemize}
\item $d(x,x) = 0$,
\item $d(x,z) \leq d(x,y) + d(y,z)$.
\end{itemize}
 A \hl{pseudometric} is a distance that also satisfies $d(x,y) = d(y,x)$.
A \hl{quasimetric} is a distance that also satisfies $d(x,y) = 0 \Rightarrow x = y$.
 A \hl{metric} is a distance that is both a quasimetric and a pseudometric.
We call a distance \hl{standard} if $d(E, E') = \infty$ whenever $E$ and $E'$ have
different sets of voters or candidates.
\end{defn}

One commonly used class of distances consists of the \hl{votewise}
distances  \cite{EFS2015} (Definition~\ref{def:votewise} below). First we require some 
preliminary definitions.

\begin{eg} (commonly used distances on $L(C)$)
\label{eg:dist Sn}
We discuss the following distances on $L(C)$ in this article.
\begin{itemize}
\item The \hl{discrete metric} $d_H$, defined by
$$d_H(\rho, \rho') =  
\begin{cases} 1 \quad \text{if $\rho = \rho'$} \\
0 \quad \text{otherwise}.
\end{cases}
$$
\item The \hl{inversion metric} $d_K$ (also called the swap, bubblesort or Kendall-$\tau$ metric), where $d_K(\rho, \sigma)$ is the 
minimum number of swaps of adjacent elements  needed to change $\rho$ into $\sigma$.
\item \hl{Spearman's footrule} $d_S$, defined by 
 $$d_S(\rho, \rho'):= \sum_{c\in C} |\rk(\rho, c) -\rk(\rho',c)|.$$
 Here $\rk(\rho, c)$ denotes the rank of $c$ in the preference order $\rho$.
 \end{itemize}
\end{eg} 

\begin{defn}
\label{def:norm}
A \hl{seminorm} on a real vector space $X$ is a real-valued function $N$ satisfying the identities
\begin{itemize}
\item  $N(x+y) \leq N(x)+N(y)$
\item  $N(\lambda x) = |\lambda| N(x)$
\end{itemize}
for all $x,y \in X$ and all $\lambda\in \reals$. Note that this implies that $N(0) = 0$ and 
$N(x) \geq 0$ for all $x\in X$.

A \hl{norm} is  a seminorm that also satisfies 
\begin{itemize}
\item $N(x) = 0 \Rightarrow x=0$.
\end{itemize}
\end{defn}

\begin{remark}
Every seminorm induces  a pseudometric via $d(x,y) = ||x-y||$. This is a metric if and only if the 
seminorm is a norm.
\end{remark}

\begin{eg}
\label{eg:norm}
Consider an $n$-dimensional space $X$ with fixed basis $e_1, \dots, e_n$ and corresponding 
coefficients $x_i$ for each element $x\in X$. Fix $p$ with $1\leq p < \infty$ and define the 
$\ell^p$-norm on $X$ by 
$$
||x||_p = \left(\sum_{i=1}^n |x_i|^p\right)^{1/p}.
$$
When $p = \infty$ we define the $\ell^\infty$ norm by 
$$
||x||_\infty =  \max_{1\leq i \leq n} |x_i|.
$$
\end{eg}

\begin{defn} (votewise distances)
\label{def:votewise}

Fix a candidate set $C$ and voter set $V$, and a distance $d$ on $L(C)$. Choose a family $\{N_n\}_{n\geq 1}$ of
seminorms, where $N_n$ is defined on $\reals^n$. Extend $d$ to a
function on $\profs(C,V)$ by taking $n = |V|$ and defining for $\sigma, \pi
\in \profs(C,V)$ 
$$d^{N_n}(\pi,\sigma):= N_n(d(\pi_1, \sigma_1), \dots, d(\pi_n, \sigma_n)). $$ 
This yields a distance on elections having the
same set of voters and candidates. We complete the definition of the extended distance
(which we denote by $d^N$) on $\elecs$ by declaring it to be standard.

We use the abbreviation $d^p$ for $d^{\ell^p}$, and sometimes we even
use just $d$ for $d^N$ if the meaning is clear.
\end{defn}

\begin{remark}
Note that if $d$ is a metric and $N$ is a norm, then $d^N$ is a metric. 
\end{remark}

\begin{eg} (famous votewise distances)
\label{example:dist}
The  distances $d^1_H$ and $d^1_K$ are 
called respectively the \hl{Hamming metric} and \hl{Kemeny metric}.
The Hamming metric measures the minimum number of voters whose preferences 
must be changed in order to convert one profile to another, and as such has an interpretation in 
terms of \emph{unit cost bribery} \cite{FaHe2009}. The Kemeny metric measures the minimum number of 
swaps of adjacent candidates  
required to convert one profile to another, and is related to models of voter error \cite{Youn1995}.
\end{eg}

\begin{eg} (tournament distances)
\label{example:dist 2}
Given an election $E= (C, V, \pi)$, we form the \hl{pairwise tournament digraph} $\Gamma(E)$ with
nodes indexed by the candidates, where the arc from $a$ to $b$ has
weight equal to the \emph{net support} for $a$ over $b$ in a pairwise contest.
Formally, there is an arc from $a$ to $b$ whose weight equals $n_{ab} - n_{ba}$, where 
$n_{ab}$ denotes the number of rankings in $\pi$ in which $a$ is above $b$.

Let $M(E)$ be the weighted adjacency matrix of $\Gamma(E)$ (with respect to an arbitrarily chosen
fixed ordering of $C$). Given a seminorm $N$ on the space of all $|C| \times |C|$ real matrices, 
we define the $N$-\hl{tournament distance} by 
$$
d^N(E, E') = N(M(E) - M(E')).
$$  

A closely related distance is defined in a
similar way, but where each element of the adjacency matrix is
replaced by its sign ($1$, $0$, or $-1$).  We call this the
$N$-\hl{reduced tournament distance}. We denote the special cases
where $N$ is the $\ell^1$ norm on matrices by $d^T$ and $d^{RT}$ respectively. Every (reduced)
tournament distance is a pseudometric.
\end{eg}

\if01

\begin{eg} (a strange pseudometric)
\label{eg:weird dist}
Let $R$ be a social rule. Define $d$ as follows. First, $d(E, E') = 0$ if and only if $R(E)=R(E')$ and 
$|R(E)|=1$; that is, both elections have the same unique winner under $R$. Second, 
$d(E, E') = 1$ if and only if exactly one of $R(E)$ and $R(E')$
 (without loss of generality, $E$) is a singleton, and $R(E) \subset R(E')$. In other words, the 
winner at $E$ is contained in the set of winners at $E'$. Finally, define $d(E,E') = 2$ otherwise.
 
We claim that $d$ is a pseudometric. It is symmetric by definition, and has the correct domain 
and codomain. Furthermore $d(E,E) = 0$ no matter what the value of $|R(E)|$. It remains to check the triangle inequality. Suppose that $d(E, E'') > d(E, E') + d(E',E'')$. If any two of $E, E', E''$ have 
the same winner set under $R$, then all have the same winner set, in which case all three distances have value 0 or all three have value 2, yielding a contradiction. If all winner sets are different, then $d(E, E')$ and $d(E', E'')$ both have value at least 1, again yielding a contradiction.
\end{eg}

\begin{eg} (quasimetrics)
Quasimetrics occur in situations when there is asymmetry in the cost of changing a vote. 
For example, it may be much more costly (for social reasons) to change a profile away from
unanimity than towards it. Some rules, for example Young's rule, are defined in terms of deletion
of voters and for this nonstandard distances are needed. For example,
let $d_{del}(E,E')$ (respectively $d_{ins}(E,E')$) be defined as the
minimum number of voters we must delete from (insert into)  election $E$
in order to reach election $E'$ (or $+\infty$ if $E'$ can never be
reached). Each is nonstandard and a quasimetric. Their symmetrized versions, which are metrics \cite{EFS2012}, are still nonstandard.
\end{eg}

\begin{eg} (shortest path distances)
\label{eg:geodesic}
We can use $\elecs$ as the underlying set of a digraph, by 
defining an arc between $E$ and $E'$  if and only if $d(E, E') = 1$. 
Then the length of a shortest path on such a digraph gives a quasimetric $d'$, which is a metric if 
the underlying digraph is a graph. Sometimes, $d' = d$. For example, $d_H, d_K, d_{ins}, d_{del}$ 
are essentially defined via this construction. 
\end{eg}

\fi

\if01

While shortest path distances occur often in the literature, they are rather special.
\begin{prop}
\label{prop:short path}
A quasimetric on $\elecs$ is a shortest path distance for some nonempty
edge relation on $\elecs$ if and only if it takes integer values, and
for each $y, x\in \elecs$ such that $2\leq d(y,x) < \infty$, there is
$z\in \elecs, z\neq x, z\neq y$ such that $d(y,z) + d(z,x) = d(y,x)$.

\end{prop}
\begin{proof}
Each shortest path quasimetric satisfies the given conditions. For the
converse, suppose that $d$ is a quasimetric on $\elecs$ satisfying the
given conditions. It follows that for every $y,x$, there are $x_0=x,
\dots, x_k=y$ such that $d(y,x) = k$ and each $d(x_i, x_{i+1}) = 1$.
Define an arc between $E$ and $E'$ if and only if $d(E,E') = 1$. Let
$d'$ be the shortest path distance on the associated digraph. It follows
by induction on the minimum value of $k$ that $d = d'$.
\end{proof}
Note that if we allowed arbitrary nonnegative weights on the edges of the digraph, every distance 
could be represented as a shortest path distance. 

\fi

\subsection{Combining consensus and distance}
\label{ss:combine}

In order for a rule to be definable via the DR construction, it is
necessary that the following property holds, and we shall assume this
from now on.

\begin{defn} \label{def:distinguish}
Let $d$ be a distance and $\cons$ a consensus. We say that $(\cons, d)$
\hl{distinguishes consensus choices} if whenever $x\in \cons_r, y \in \cons_{r'}$
and $r\neq r'$, then $d(x,y) > 0$.
\end{defn}

We use a distance to extend a consensus to a social rule in the natural
way. The choice at a given election $E$ consists of all $s$-rankings $r$
whose consensus set $\cons_r$ minimizes the distance to $E$. We
introduce the idea of a score in order to use our intuition about
positional scoring rules. 

\begin{defn} (DR scores and rules)
\label{definition:rules}

Suppose that $\cons$ is an $s$-consensus and
$d$ a distance on $\elecs$. Fix an election $E\in \elecs$. For each $r\in L_s(C)$, the
\hl{$(\cons, d, E)$-score} of $r$ is defined by
$$
|r| : = d(E, \cons_r) := \inf_{E'\in \cons_r} d(E, E').
$$

The rule $R:=\R(\cons, d)$ is defined by
\begin{equation}
\label{eq:argmin}
R(E) = \arg\min_{r\in L_s(C)} |r|.
\end{equation}

We say that $R$ is \hl{distance rationalizable} (DR) with respect to
$(\cons, d)$.
\end{defn}

\if01
\begin{remark}
Note that if $\cons_r$ is empty, then $|r| = \infty$. DR scores are defined so that they are 
nonnegative, and higher score corresponds to larger distance. This is not consistent with
the usual scoring rule interpretation in Example~\ref{eg:scoring}, but the two notions of score 
 are closely related. Our DR scores have the form $M - s$ where $s$ is the score
associated with the scoring rule and $M$ depends on $E$ but not on any $r\in
L_s(C)$.
\end{remark}
\fi

\if01
Table~\ref{t:DR egs} presents a few known rules in this framework. Most of the rules in the 
table are well known. We single out the following less obvious references. The \hl{modal ranking rule} is a very natural rule described in  \cite{CPS2014}. The \hl{voter replacement rule} 
(VRR) was 
defined essentially as a missing entry in such a table \cite{EFS2010}. 
The  entries marked ``trivial" are so labelled because in those cases every election not in 
$\cons$ is at distance $+\infty$ from every $\cons_r$. Missing entries reflect on the authors' 
knowledge, 
and may have established names. Our table overlaps with that in \cite{MeNu2008} --- note that the 
$(\cond, d_H^1)$ entry is incorrect in that reference, as pointed out by Elkind, Faliszewski 
and Slinko \cite{EFS2012}.

\begin{table}
\begin{tabular}{c|cccc}

$\cons / d$ & $\sunam$ & $\wunam$ & $\cond$ & $ \cond^m$  \\
\hline
$d_K^1$ & Kemeny & Borda & Dodgson & \\
$d_H^1$ &modal ranking& plurality & VRR & \\
$d_S^1$ &Litvak &Borda&Dodgson &\\
$d_T$ &Kemeny&Borda&maximin& \\
$d_{RT}$ &Copeland &Copeland & Copeland & Slater\\
$d_{ins}$ &trivial &trivial & maximin  & \\ 
$d_{del}$ &modal ranking & plurality &Young & \\
\end{tabular}
\caption{Some known rules in the DR framework}
\label{t:DR egs}

\end{table}
\fi

\begin{eg} (scoring rules)
\label{eg:scoring}
The \hl{positional scoring rule} defined by a vector $w$ of \hl{weights} with $w_1
\geq \dots \geq w_m$ and $w_1 > w_m$  elects all candidates with maximum score, where
the score of $a$ in the profile $\pi$ is defined as $\sum_{v\in V}
w_{\rk(\pi(v), a)}$. 

The \hl{plurality} ($w = (1,0,0, \dots, 0)$) and \hl{Borda} ($w = (m-1,m-2, \dots, 0)$) 
rules are well-known special cases. 

The positional scoring rule defined by $w$ has the form $\R(\wunam,
d_{S(w)}^1)$ where $d_{S(w)}$ is the generalization of $d_S$ given by $$d_{S(w)}(\rho,
\rho') = \sum_{c\in C} |w_{\rk(\rho,c)} - w_{\rk(\rho',c)}|.$$ \cite[Prop. 8]{EFS2015}. Thus the Borda rule can be expressed as $\R(\wunam,
d_S^1)$.

Also, the positional scoring rule defined by $w$ has the form $\R(\wunam, d_{K(w)}^1)$ where $d_{K(w)}$ is a generalized swap distance on rankings \cite{LeNi1985}. Borda's rule has the form $\R(\wunam, d_K^1)$.

\end{eg}

\begin{remark}
Note that $d_w$ is a metric on $L_s(C)$ if and only if $w_1, \dots, w_s$
are all distinct. 
The score of $r$ under the rule defined by $w$ is the 
difference $nw_1 - |r|$. For example, for Borda with $m$ candidates the
maximum possible score of a candidate $c$ is $(m-1)n$, achieved only for
those elections in $\wunam_c$. Note that as far as the distance to $\wunam$ or $\cond$ is
concerned, $d_S^1$ and $d_K^1$ are proportional, but 
they are not proportional in
general \cite[p. 298--299]{MeNu2008}. The score of $c$ under Borda is exactly
$n(m-1) - K$ where $K$ is the total number of swaps of adjacent
candidates needed to move $c$ to the top of all preference orders in
$\pi(E)$. 
\end{remark}

\if01
\begin{remark}
Consider the set of all elections for which all
positional scoring rules yield the same unique winner. By convexity of
the set of weight vectors, it suffices to check the finite set of rules
defined by weight vectors $(1, 1, \dots, 1, 0, \dots , 0)$ where the
number of $1$'s is fixed and at most $k-1$; in other words the \hl{$k$-approval rules} for $1\leq
k \leq m-1$. This just describes the Lorenz consensus in another way.
\end{remark}
\fi

\begin{eg} (Copeland's rule)
\hl{Copeland's rule} can be represented as $\R(\cond, d_{RT})$. Indeed,
in an election $E$, the Copeland score of a candidate $c$  (the number
of points it scores in pairwise contests with other candidates) equals
$n-1-s$, where $s$ is the minimum number of pairwise results that must be changed 
for $E$ to change to an election that belongs to $\cond_c$.
\end{eg}


\section{Simplex rules}
\label{s:simplex}

Although it is far from the general case, most rules used in practice are in fact anonymous and homogeneous. In this case there is an appealing geometric interpretation. The frequency distribution of votes is sufficient information to determine the output of the rule, and so profile space can be substantially compressed. 

We have previously explored in detail the connection with distance rationalization  \cite{HaWi2016}. We first recall the construction for anonymous rules. Recall that a \hl{multiset} of weight $n$ on an underlying set $S$ of size $M$ is ``a set of $n$ elements of $S$ with repetitions". Formally, there is a function $f:S\to\nats$ where $f(s)$ gives the multiplicity of $s$ in the multiset.

\begin{defn}
Let $E = (C,V, \pi) \in \elecs$.
The \hl{vote number map} $\nummap$ is the map that 
associates $E$ with the multiset $\nummap(E)$ on $L(C)$ of weight $|V|$, in which the multiplicity of $\rho\in L(C)$ is the number of voters in $V$ having that preference order.

A rule $R$ is \hl{anonymous} if 
$R(E) = R(E')$ whenever $\nummap(E) = \nummap(E')$.
We denote the quotient space by $\votesits$ and call it the set of \hl{anonymous profiles}. 
\end{defn}

\begin{remark}
$\nummap(E)$ simply keeps track of the numbers of votes of each type in $\pi$, ignoring the identities of voters. A rule is anonymous if this information is enough to determine the output. The more usual (and equivalent) definition of anonymity is given in Remark~\ref{rem:anon} below. 
\end{remark}

\begin{defn}
\label{def:distmap}

The \hl{vote distribution} associated to $E$ is the  relative frequency distribution on $L(C)$ corresponding to the multiset $\nummap(E)$, which we denote $\distmap(E)$. 
Explicitly, $\distmap$ is a function from $\elecs$ to $[0,1]^{L(C)}$, which gives for each ranking the proportion of voters having this ranking as preference. 
The vote distribution map defines an equivalence relation $\sim$ on $\elecs$ in the usual way: $E\sim E'$ if and only if $\distmap(E) = \distmap(E')$.

An anonymous rule $R$ is \hl{homogeneous} if 
$R(E) = R(E')$ whenever $\distmap(E) = \distmap(E')$.
We denote the 
quotient space by $\votedists$, and call it the set of \hl{anonymous and homogeneous profiles}.
\end{defn}

\begin{remark}
Note that if $E = (C, V, \pi), E' = (C, V', \pi') \in \elecs$ and $\nummap(E) = \nummap(E')$, then $|V| = |V'|$. Thus there is a well-defined map $f:\votesits \to \votedists$ (``divide by the number of voters"), and $\distmap(E)$ simply lists each preference order according to its relative frequency in $\pi$.
\end{remark}

\begin{remark}
\label{rem:anon}
Anonymous and homogeneous rules were already defined in the literature, but in different ways.

Let $g$ be a permutation of $V^*$. For each $E = (C, V, \pi)$, define $g(E) = (C, g(V), g(\pi))$ where $g(\pi)(v) = \pi(g(v))$. A rule $R$ is anonymous if and only if $R(g(E)) = R(E)$ for all $E\in \elecs$ and all $g$.

For $k\geq 1$, define $kE$ to be an election $(C, kV, k\pi)$ where $kV$ consists of $k$ copies of each voter in $V$ and $k\pi$ the corresponding copies of their preference orders (the exact order of the voters is irrelevant since we deal only with anonymous rules). An anonymous rule $R$ is homogeneous if and only if  $R(kE) = R(E)$ for all $E\in \elecs$ and all $k$.

\end{remark}

We call anonymous and homogeneous rules \hl{simplex rules} for short, and now explain why. So far the discussion has been coordinate-free, but it is often useful to introduce coordinates. Given any linear ordering $\rho_1, \rho_2, \dots, \rho_M$ 
on $L(C)$, where  $|C|=m$ and $M=m!$, we can introduce coordinates $x_i$ such that $x_i$ denotes the relative frequency associated to $\rho_i$. The set of frequency distributions 
$\simp^\rats(L(C))$ is then coordinatized by the rational points of the standard simplex.

\begin{defn}
The \hl{standard simplex} in $\reals^M$ is the set
$$
\simp_M:= \left\{x \in \reals^M \mid \sum_i x_i = 1, x_i \geq 0 \text{ for $1\leq i \leq M$} \right\}.
$$
We let $\simp^{\rats}_M:=\rats^M \cap \simp_M$ denote the rational points of $\simp_M$.
\end{defn}

\begin{remark}
For simplicity we sometimes write $x_t$ for the component of $x\in \simp_M$ corresponding to $t\in L(C)$ (instead of $x_i$ where $i$ is the number of $t$ in some linear ordering on $L(C)$).
\end{remark}

\begin{eg}
\label{eg:simp}
An election on candidates $C = \{a,b,c\}$ having $7$ voters of whom $3$ have preference $a\succ b\succ c$, $2$ have preference $b\succ a \succ c$ and $2$ have preference $c\succ b\succ a$ corresponds (under the lexicographic order on $L(C)$) to the anonymous profile $(3,0,2,0,0,2)$ and hence to the point $(3/7,0,2/7,0,0,2/7)\in \simp_6$.
\end{eg}

\begin{remark}
We always consider $\simp_M$ as embedded in $\reals_M$. It is contained in a unique hyperplane $H_M$, which is given by the single linear equation $\sum_{i=1}^M x_i = 1$.
\end{remark}

We can interpret each anonymous and homogeneous rule as being defined on $\simp^{\rats}_M$. If the rule is also \emph{continuous}  then we can in fact define it on $\simp_M$, which allows us to use our usual geometric intuition. For example, the domain of 
$\sunam$ consists of the corners of the simplex, while the domain of $\wunam$ also lies on the boundary of the simplex (for example, 
$\wunam_a$ contains all points of the simplex for whom all coordinates corresponding to rankings with $a$ not at the top are zero). 

The general approach of the last paragraph has been adopted by many previous authors. For example, Saari \cite{Saar1995} simply uses the terminology ``profiles" to refer to vote distributions, and all rules he considers are continuous simplex rules by definition. 

\subsection{Distance rationalization in the simplex}
\label{ss:DR simp}

We shall see in Section~\ref{ss:quot dist} how anonymous and homogeneous distance rationalizable rules defined using profiles can be interpreted using the simplex. 
The converse idea is to define distance rationalizable rules directly on the simplex rather than on profile space. We make the obvious definitions by analogy with those for profiles.

\begin{defn}
\label{def:DR simp}
Given a fixed candidate set $C$ of size $m$ and a distance on $\simp^\rats_M$ where $M=m!$, a \hl{partial social rule} on $\simp^\rats_M$ of size $s$ with domain $D\subseteq  \simp^\rats_M$ is a mapping taking each element of $D$  to a nonempty subset of $L_s(C)$. A \hl{consensus} on $\simp^\rats_M$ is a partial social rule that is single-valued at every point (a partial social function). Given a consensus $K$ and distance $\delta$ on $\simp^\rats_M$, the rule $\R(K, \delta)$ is defined by 
\begin{equation}
\label{eq:DR argmin}
R(x) = \arg\min_{r\in L_s(C)} \delta(x, K_r).
\end{equation}
\end{defn}

The most obvious distances mathematically are surely the $\ell^p$ metrics. 
The interpretation in terms of social choice is less compelling for $p > 1$, since we are measuring the amount of effort needed to 
change one election into another by transferring  vote mass under a nonlinear penalty. The case $p = 1$ is by far the most commonly studied, and also arises directly from votewise distances, unlike the case $p>1$.

To our knowledge, several fairly obvious rules of this type have not yet been studied in detail. Here is an example using the unanimity consensus.

\begin{prop}

\label{prop:$ell^p$}
Fix $p$ with $1\leq p \leq \infty$ and consider the social choice rule $\R(\sunam^s, \ell^p)$ defined on $\simp^\rats_M$. This rule chooses precisely the initial $s$-ranking from each of the most frequent ranking(s) from the input profile.
\end{prop}

\begin{proof}
Let $\delta$ be the $\ell^p$ distance on $\reals^M$ and let $e_t$ denote the basis vector in $\reals^M$ corresponding to $t\in L(C)$, a corner of the simplex.  Then for $x\in \Delta_M$ and $t, t'\in L(C)$,
\begin{align*}
\delta^p(x, e_t) - \delta^p(x, e_{t'}) & =  (1-x_t)^p + \sum_{k\neq t} x_t^p - (1-x_{t'})^p - \sum_{k\neq t'} x_{t'}^p\\
& = (1-x_t)^p - (1-x_{t'})^p + x_{t'}^p - x_t^p.
\end{align*}
Choosing $x_{t^*}$ to be the maximum value among the entries of $x$ and setting $t=t^*$ shows the right side of the above expression to be nonpositive for all $t'\neq t^*$. Hence $t^*$ is a minimizer (the same argument works for $p = \infty$ with a different computation). 

Thus if $\rho^*$ denotes the initial $s$-ranking corresponding to $t^*$ (written $t^*_s = \rho^*$), then for each $s$-ranking $\rho$, we have
$$
\delta(x, \sunam^s_{\rho}) = \min_{\{t: t_s = \rho \}} \delta(x, e_t) \geq \delta(x, e_{t^*}) = \delta(x, \sunam^s_{\rho_*}).
$$
\end{proof}

\begin{remark}
For $s=m$ this rule is 
the \hl{modal ranking rule} \cite{CPS2014} (the term \hl{plurality ranking rule} may seem more logical, but it is important not to be confused with the ranking induced by plurality scores of candidates in the case where these scores all differ). For $s=1$ it is the social choice variant of the modal ranking rule, choosing the highest ranked candidate from each modal ranking. Note that the latter rule differs from plurality rule, which is what we would get if we used the simplex in dimension $m$, as for example done by Saari \cite{Saar1995}.
\end{remark}

In order to find more interesting rules defined naturally on $\simp^\rats_M$, we may need to use less common distances. 
For example, so-called \emph{statistical distances} such as the (asymmetric) \emph{Kullback--Leibler divergence} (or \emph{relative entropy}) and the \emph{Hellinger metric} are heavily used in many application areas (the simpler \emph{total variation distance} arises in Example~\ref{eg:transport} below). We do not present a detailed study here, deferring it to future work. Instead, we now move on to explore distances on $\simp^\rats_M$ induced by distances on $\elecs$.

\subsection{Quotient distances}
\label{ss:quot dist}

We can define a simplex rule by first starting with profiles and passing to the quotient space, provided the
rule in question is anonymous and homogeneous (the case of Dodgson's rule $\R(\cond, d_K^1)$ which is anonymous but not homogeneous shows that care must be taken). 

Since votewise distances are 
very natural and the $\ell^1$ norm is the most obvious 
choice for a votewise distance (because it just adds the distance from each voter), we obtain several rules with an $\ell^1$ flavour in this way. For example, the Hamming metric yields a constant multiple of $\ell^1$ via the Wasserstein construction as described in Example~\ref{eg:transport} below, while the Kemeny metric and $\ell^1$ lead to rules such as Borda and Kemeny's rule. We discuss $\ell^1$-votewise rules in more detail in Section~\ref{s:ell one}.

\begin{defn}
\label{def:anon dist}
A distance is \hl{anonymous} if  $d(g(E), g(E')) = d(E, E')$ for all $E, E'\in \elecs$ and all permutations $g$.
An anonymous distance is \hl{homogeneous} if $d(kE, kE') = d(E,E')$ for all $E, E'\in \elecs$ and all $k\geq 1$.
\end{defn}
Note that votewise distances need to be normalized before becoming homogeneous, as we explain in Remark~\ref{rem:normalize}. 

Every anonymous and homogeneous rule $R$ corresponds to a simplex rule $\overline{R}$. When  $R = \R(\cons, d)$ where both $\cons$ and $d$  are anonymous and homogeneous, we can express $\overline{R}$ as $\R(K, \delta)$ where $K, \delta$ depend nicely on $\cons, d$. The mapping from profiles to the simplex yields the obvious consensus $K = \overline{\cons}$. The distance $\delta$ is a little more involved. The obvious idea is to use a quotient distance \cite{DeDe2009}. This concept is standard but relatively little-known.

\begin{defn}
\label{def:quot dist}
We define $\overline{d}: \simp_M \times \simp_M \to \reals_+$ to be the \hl{quotient distance}
induced by $\sim$. That is,
\begin{equation}
\label{eq:quot}
\overline{d}(x,y)  = \inf \sum_{i=1}^k d(E_i, E'_i)
\end{equation}
where the infimum is taken over all  \emph{admissible paths}, namely
paths such that $E'_i \sim E_{i+1}$ for $1\leq i \leq k-1$, $E=E_1, E' = E'_k$, $E$ projects
to $x$ and $E'$ to $y$. 
\end{defn}

In general, quotient distances are tricky to work with, owing to the complicated definition. In our setup, it turns out that they are reasonably tractable.

\begin{prop}
\label{prop:homog simple}
Let $d$ be an anonymous and homogeneous standard distance 
on $\elecs$. Then $\overline{d}$ is given by 
$$
\overline{d}(x, y) = \inf_{E'} d(E, E')
$$
where $E, E'$ range over all elections having an equal number of voters, such that $\distmap(E) = x, \distmap(E') =y$.
\end{prop}

\begin{proof}
Let $x,y\in \votedists$ and consider an admissible path and its corresponding sum
$$
\sum_{i=1}^k d(E_i, E'_i)
$$
where $k>1$. We show that we can reduce the value of $k$.
By homogeneity of $d$, we can change the size of the voter sets so that $E_k$ and $E'_{k-1}$ have the same sized voter set. Then we 
can choose a permutation $g$ of $V^*$ taking $E_k$ to $E'_{k-1}$. Since $d$ is anonymous, using the triangle inequality we obtain
\begin{align*}
d(E_{k-1}, E'_{k-1}) + d(E_k, E'_{k}) &= d(E_{k-1}, E'_{k-1}) + d(g(E_k), g(E'_{k})) \\ & =  d(E_{k-1}, E'_{k-1}) + d(E'_{k-1}, g(E'_k)) \\
&\geq d(E_{k-1}, g(E'_k)).
\end{align*}
This gives an admissible path with $k$ replaced by $k-1$. 

Thus without loss of generality, in computing $\overline{d}$ we need only deal with the case $k=1$. Finally, if the minimum is reached for some $E=E_*$, then for each $E$ we can write $g(E_*) = E$ for some $g$. Thus by anonymity $d(E, E') = d(E_*, g^{-1}(E'))$. Thus the minimum can be taken only over $E'$, since $g^{-1}(E')$ ranges over all elections as $g$ ranges over all permutations. 
\end{proof}

\begin{remark}
Proposition~\ref{prop:homog simple} shows that we can replace $\inf$ by $\min$, since by anonymity, once we fix the number of voters there are only a finite number of possible distances $d(E, E')$ to consider.
\end{remark}

\begin{prop}
Let $d$ be an anonymous and homogeneous standard distance 
on $\elecs$, and $K$ be an anonymous and homogeneous consensus.
If $R=\R(K, \delta)$ is anonymous and homogeneous, then  $\overline{R}=\R(\overline{K},\overline{d})$.
\end{prop}
\begin{proof}
By Definition \ref{definition:rules}, $R(E) = \arg\min_{r\in L_s(C)} d(E, K_r)$. By Proposition~ \ref{prop:homog simple}, $\overline{d}(x, y) = \inf_{E'} d(E, E')$, such that $\distmap(E) = x, \distmap(E') =y$. In particular, $\overline{d}(\distmap(E),\overline{K}_r)=\inf_{E'} d(E, E')$ with $E'\in K_r$, so $\overline{d}(\distmap(E),\overline{K}_r)=d(E,K_r)$. Thus, $\overline{R}(\distmap(E))=R(E)=\arg\min_{r\in L_s(C)} \overline{d}(\distmap(E),\overline{K}_r)$ and $\overline{R}=\R(\overline{K},\overline{d})$.

\end{proof}


\subsection{The $\ell^p$-votewise case --- Wasserstein distance}
\label{ss:wasserstein}

In the special case of $\ell^p$-votewise distances, we can describe $\overline{d}$ in more detail
using a well-known construction from probability theory, the Wasserstein distance, which we now recall.

Let $S$ be a finite set of size $M$, and let $\simp(S)$ denote the set of probability distributions on $S$.
For a distance $d$ defined on $S$, the function $d^p_W:\simp(S)\times\simp(S)
\to \reals$ is defined by \[d^p_W(x,y)^p=\inf_A\sum_{r,r'\in S}A_{r,r'}d(r,r')^p,\] 
where the infimum is taken over all \emph{couplings} of
$x$ and $y$, defined as nonnegative square matrices of size $M$ whose
marginals are $x$ and $y$ respectively (i.e. $\forall r,
\sum_{r'}A_{rr'}=x_r$ and $\forall r',\sum_rA_{rr'}=y_{r'}$). Basically,
it represents the minimum cost to move from one configuration to
another, where the underlying distance $d$ defines the cost of each
movement. Indeed, this construction leads to a new distance.

\begin{prop}
If $d$ is a distance on $S$, then $d_W^p$ is a distance on $\simp(S)$. If $d$ is a metric, then so is $d_W^p$.
\end{prop}
\begin{proof}
See \cite[Ch. 6]{Vill2008}.
\end{proof}

\begin{remark}
The function $d_W^p$ goes by several names, some common ones being the \hl{$l^p$-transportation} distance, the \hl{Kantorovich $p$-distance}, the
\hl{$p$-Wasserstein distance}. When $p=1$, it is also called the \hl{Earth Mover's}
distance or \hl{first Mallows metric}, and is 
used heavily in several areas of computer science, particularly
image retrieval and pattern recognition.
\end{remark}

Now we are able to make the link with the votewise metrics, by 
applying the construction of $d^p_W$ in the case $S = L(C)$. Scaling a votewise distance based on the $\ell^p$ norm gives a homogeneous distance with a special formula, and this turns out to be exactly the $p$-Wasserstein distance.

\begin{remark}
\label{rem:normalize}
A votewise distance based on $\ell^p$ is not homogeneous, as its value depends on the number of voters. However, we may make an equivalent homogeneous version by scaling. For each real number $p$ with $1\leq p \leq \infty$, and each positive integer $n$, we define a new norm on $\mathbb{R}^n$ by defining for each $x\in \mathbb{R}^n$ 
$$||x||^*_p = \frac{1}{n}||x||_p.$$
Denote the votewise distance corresponding to this norm by $d_*^p$; then  $\R(\cons, d^p) =  \R(\cons, d_*^p)$.
\end{remark}

\begin{prop}
Let $d$ be a finite distance on $L(C)$. Then $\overline{d_*^p} = d^p_W$.
\end{prop}

\begin{proof}
Write $S = L(C)$. Since $d_*^p$ is anonymous and homogeneous, Proposition~\ref{prop:homog simple} implies that $\overline{d_*^p}(x,y)=\min_{E,E'}d_*^p(E,E')$, where
$\overline{E} = x$, $\overline{E'}=y$ and $|V(E)| = |V(E')|$.
Fix $n\geq 1$ and let $E=(C,V,\pi)$ and $E'=(C,V,\pi')$ denote elections with $n$ voters such that $\overline{E} = x$, $\overline{E'}=y$.
Then
\[\overline{d_*^p}(x,y)^p=\min_{\pi'}\frac{1}{n}\sum_i
d(\pi_i,\pi'_i)^p \geq \frac{1}{n}\sum_{r,r'\in S} \sum_{\{i:\pi_i=r,\pi'_i=r'\}}
d(r,r')^p=\sum_{r,r'\in S}\frac{a_{r,r'}}{n}d(r,r')^p,\]
where the $a_{r,r'}$'s are nonnegative integers such that for all $r\in
S,\sum_r\frac{a_{r,r'}}{n}=x_r$ and for all $r'\in
S,\sum_{r'}\frac{a_{r,r'}}{n}=y_{r'}$, which corresponds to the
Wasserstein distance restricted to matrices $A$ respecting the conditions
and with coefficients of the form $\frac{k}{n}$ with $0\leq k\leq n$ (in other words, a rational coupling with restricted denominators). So
clearly, $d^p_W(x,y)\leq\overline{d_*^p}(x,y)$.

Let assume that this inequality is strict. Then there is a coupling $A'$ (not all of whose entries are rational) with 
$\sum_{r,r'\in S}a'_{r,r'}d(r,r')^p + \epsilon/2 < \sum_{r,r'\in S}a_{r,r'}d(r,r')^p$ for all rational couplings $A$. However since $\max_{r,r'}d(r,r')<\infty$ and we can choose $n$ as large as we want, $d^p_W(x,y)$
can be approximated arbitrarily closely, since we can approach all entries of $A'$ simultaneously arbitrarily closely by a rational matrix satisfying the coupling constraints. This contradiction yields the final result.
\end{proof}

\begin{eg}
\label{eg:transport}
Let $d_1$ be the $\ell^1$ distance between probability measures on $L(C)$ and let $d = d_H$. 
Then $\overline{d_*^1} = \frac{1}{2} d_1$ (also called the \hl{total variation distance}). 
This was observed
(without the current notation) in  \cite[Lemma 3.2]{MPR2013}: if $E, E'$
are elections on $(C,V)$ with $|V| = n$, then 
$$
d_1(\overline{E}, \overline{E'}) \leq \frac{2}{n} d_H(E, E')
$$
and given $\overline{E}, \overline{E'} \in \simp^\rats_M$, we can choose
$n$ and the preimages $E,E'$ so that equality holds. 
\end{eg}

\begin{eg}
\label{eg:wasserstein}
If
$x\in \votesits$ has $2$ ``$abc$ voters" ( voters with preferences $a \succ b \succ c$) and $3$ $bac$ voters, while $y$ has $2$ $bac$ voters and $3$ $cba$ voters, the quotient distance corresponding to the normalized version of $d_H^1$ is $3/5$. This is because we must change at least $3$ of the $5$ voters to convert $x$ to $y$ (switch both $abc$ voters to $cba$, and one of the $bac$ voters to $cba$). For
the Kemeny metric $d_K^1$, the analogous quantity is $8/5$, because we must make at least $8$ swaps of adjacent candidates to convert $x$ to $y$ (switching $abc$ to $cba$ requires $3$ swaps, and switching $bac$ to $cba$ requires $2$ swaps; similarly, switching both $abc$ votes to $bac$ and all three $bac$ to $cba$ requires $8$ swaps).

\end{eg}

\section{Tied sets and decisiveness} 
\label{s:ties}

All social  rules used in practice encounter the problem of breaking
ties. However, many commonly used social rules have the property that
the subset of profiles where a unique output is not determined is ``small" (for example asymptotically negligible as $n\to
\infty$, for fixed $m$). This is important, because if the region where ties occur
is small, then ties can be ignored for many purposes, whereas if that
region is asymptotically large, our rule may suffer extreme lack of
decisiveness. 

In the worst case, the rule may do nothing, and simply return all possible 
$s$-rankings at every profile, making it useless. In the DR framework, this extreme indecision cannot occur, because some consensus set must be 
nonempty. However, it is certainly possible to have ``large" subsets of profile space on which a DR 
rule is not single-valued. We investigate this question for simplex rules, giving  both  positive (few ties)  
and negative (many ties) results.

\subsection{Boundaries}
\label{ss:boundary}

We formalize the concept of ``tied region" in our geometric context.

\begin{defn} \label{definition:boundary}
The \hl{boundary} of the social rule $R$ is the set of all elections $E$ at which $R(E)$ is of size at least $2$.
\end{defn}

\begin{eg} \label{example:bdry cond 2}
Suppose that $m=2$, with alternatives $a$ and $b$,  $\cons =
\cond$, and $d$ is an anonymous neutral standard distance.  The concepts
of majority winner and Condorcet winner coincide when $m=2$. When $n$ is
odd, the boundary $\R(\cons, d)$ is empty, whereas when $n$ is even, the
elections having an equal number of $ab$ and $ba$ voters are in the
boundary. 

On the other hand, for the simplex rule $\R(\overline{\cond}, \overline{d})$, the boundary is the point $(1/2, 1/2)$.
\end{eg}

Our intuition is that the boundary of a well-behaved simplex
rule should be ``small" in $\simp_M$, and be geometrically ``nice". The
relevant geometric theory is that of \emph{Voronoi diagrams}. We now
digress to review some known results, which will help develop a more refined intuition.

\subsection{Geometric background} \label{ss:geom2}

Voronoi theory is usually defined for a metric space, and most commonly for a \hl{Minkowski space} (a finite-dimensional real normed vector space). All definitions below work for an arbitrary metric space. The theory can no doubt be generalized to distances, but we do not deal with maximum generality here. Our main interest is in explaining that there are several interesting reasons why DR rules defined by $\ell^p$-votewise distances may fail to be decisive.

For a fixed set of
of \hl{sites} (subsets of the entire space), the open \hl{Voronoi cell}
of each site $X$ is defined as the set of points closer to $X$ than to any
other site. The boundaries of these cells are contained in the union of
bisectors, where a \hl{bisector} of the sites $X$ and $Y$, denoted $\bis(X, Y)$, is defined to be the set of points equidistant from
those two sites. 

Interpreting the sites as the consensus sets $\cons_r$, we see that the open Voronoi cell 
corresponding to $\cons_r$ is precisely the set on which $\R(\cons, d)$ is single-valued with 
value $r$. Also, the boundary as defined above
is just the union of boundaries of all Voronoi cells. 

We first discuss bisectors, because if these are well-behaved, so will the boundary of
the rule be. We first restrict to the nicest situation (where our intuition is strongest),
namely where the space is $\reals^n$ under the Euclidean $\ell^2$ norm. Suppose first that the sites are all distinct points. In this case for each pair of distinct sites $X$ and $Y$, the bisector $\bis(X,Y)$ is a
hyperplane normal to the line joining the points. The Voronoi cells
are therefore convex polyhedra that tile the entire space.

However, this situation is rather special. In fact 
$\bis(X,Y)$ is a hyperplane for all $X$ and $Y$ if and only if the space
is indeed Euclidean (in other words the norm is $\ell^2$)  \cite{Mann1935}.  Thus we should expect to see bisectors that
are not hyperplanes. Of course, such bisectors may still be well-behaved, for example smooth hypersurfaces. Note that $\bis(X,Y)$ is known to be homeomorphic to a hyperplane provided
the norm is strictly convex (recall that a norm is \hl{strictly convex} if its unit sphere contains no line segment)
and even sometimes when it is not
\cite{Horv2000}. This does not preclude nasty
behaviour such as two bisectors intersecting in arbitrarily complicated ways, but for norms defined algebraically, such as $\ell^p$, that does not happen. 

When the sites are not single points (in particular when they are not separated), bisectors may be poorly behaved even in $\ell^2$ (see Example~\ref{example:common end}). 
We conclude that well-behaved bisectors should not be expected in general, and we explore this in the next section.

\subsection{Large boundary}
\label{ss:large}

The most obvious way for a rule to be rather indecisive is if the underlying distance does not distinguish points well. We say that a subset of a Minkowski space (possibly defined with a seminorm) is \hl{large} if it contains an open ball, and \hl{small} otherwise.

\begin{eg} \label{example:bdry cond 3} (pseudometric)
Consider the Copeland rule whose boundary contains all points with no unique Copeland winner. This contains in particular the set where all candidates have the same Copeland score, for example because the majority tournament contains a cycle that includes all candidates. This is a large subset of $\simp_M$. To see this, note that in terms of coordinates, the  majority cycle is described by $\binom{m}{2}$ equations of the form $\sum_{i\in S_{ab}} x_i > \sum_{i\in S_{ba}} x_i$, where $S_{ab}$ denotes the set of rankings for which $a$ is above $b$. The tied set contains a sufficiently small neighborhood of every point for which all inequalities are strict, because small changes to the proportions of voter types will not change the majority tournament. 
\end{eg}

In view of this example, we should require our distances to be quasimetrics. However, there are other more subtle problems that can occur. Note that in the next example, the distance is $\ell^2$ and is hence as nice as could be expected: a metric induced by a norm that is strictly convex, symmetric, and algebraically defined. The problem is that the consensus notion is wrong --- intuitively, consensus sets should be separated, because otherwise how could the consensus choice be uncontroversial? 

\begin{eg}(non-separated consensus)
\label{example:common end}
Let $m=3$ and consider $\simp_6$ with the usual $\ell^2$ metric (induced from $\reals^{5}$ which coordinatizes $H_6$ in the usual way).

Consider the half-open line segments $L_1, L_2$ that join the center $P$
of $\simp_6$ to the points $x_{abc} = 1$  and $x_{bca} = 1$ respectively (Figure~\ref{fig:simp} gives some intuition in lower dimension). Each contains the endpoint on the boundary, but neither contains the center of the simplex.

Define a $1$-consensus by letting $\cons_a = L_1, \cons_b = L_2$ and $\cons_c$ be the single point
$x_{cab} = 1$.  Let $H_1, H_2$ be the hyperplanes (in $\simp_6$, i.e. having dimension $4$) normal to $L_1, L_2$ at $P$. Let $S$ be the set of points in $\simp_6$ that lie 
on the other side of $H_1$ from $L_1$ and on the other side of $H_2$ from $L_2$. Then each point of $S$ lies
on the bisector of $L_1$, $L_2$ with respect to the usual
$\ell^2$ metric $d$, because the closest point of $L_1$ is $P$ and this is the closest point of $L_2$. Every point of $S$ that is closer to the centre of
$\simp_6$ than to $\cons_c$ (this includes the point $L_1\cap L_2$) is in the boundary of the social choice rule $\R(\cons, d)$,
which is therefore large. Note that $\cons$ satisfies 
anonymity and homogeneity,  but not neutrality. Furthermore, every $\cons_r$ is a convex 
polyhedral subset of the simplex. 

If instead we define $\cons$ symmetrically by letting $\cons_c$ be the line segment $L_3$ joining the point $x_{cab} = 1$ to the centre of the simplex, then although the bisectors are large on $\simp_6$, the boundary of the rule is small. This is because points in $S$ are now closer to $L_3$ than either $L_1$ or $L_2$. 

Also note that if the line segments $L_1$ and $L_2$ did not approach arbitrarily closely, the bisectors would all be small.

\end{eg}

\begin{figure}
\caption{Illustration for Example~\ref{example:common end}.}
\label{fig:simp}

\begin{tikzpicture}[scale=5]
\draw (0,1) node[anchor=east]{$a$} -- (1,0) node[anchor=west]{$b$}-- (0,0)node[anchor=east]{$c$} -- cycle;
\filldraw[fill=black!20!white] (0.1667,0) -- (0.3333,0.3333)  -- (0,0.1667) -- (0,0)  -- cycle;
\draw[very thick] (0,1) -- (0.3333,0.3333);
\draw[very thick] (1,0) -- (0.3333,0.3333);
\draw[very thick] (0,0) node{};
\end{tikzpicture}
\end{figure}
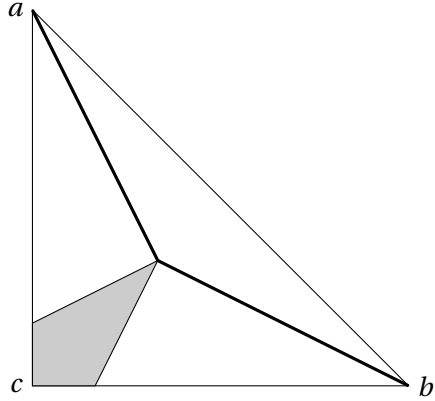

Thus we should require that consensus sets be separated. However, there is another common way in which bisectors can fail to behave well, which is when the underlying norm in a Minkowski space is not strictly convex. We now analyse a special case of this in some detail.

\section{Analysis of $\ell^1$-votewise metrics}
\label{s:ell one}

Votewise distances based on the $\ell^1$ norm are very commonly used. They are typically computationally easy and have a clear interpretation in terms of adding distances corresponding to each voter. In fact we are not aware of a named $\ell^p$-votewise rule that has been defined for any $p\neq 1$. However, when we consider decisiveness, there are some potential negative consequences to using the $\ell^1$ norm. 

We first show in
Proposition~\ref{pr:norm} and Proposition~\ref{pr:norm2} that each
$\ell^1$-votewise metric corresponds to a (Wasserstein) distance on the simplex that is induced by a norm that is not strictly convex. 

We fix a candidate set $C$ with $|C| = m$ and let $M=m!$ as in previous sections.
Let $c = (1,1,\dots, 1)/M \in \reals^M$ be the center of the simplex $\Delta_M$.
Now, we translate the center $c$ to the origin, and we denote by
$\Delta'$ the image of the simplex under this translation. We denote by
$\mathcal{H}$ the hyperplane containing $\Delta'$. Our study of the
geometry under the Wasserstein distance will be facilitated by the
following observations.

\begin{prop}
\label{pr:norm}
Let $d$ be an $\ell^1$-votewise distance. Then $d$ induces a norm $N$ on $\mathcal{H}$.
\end{prop}

\begin{proof} 
It is a well known property of the Wasserstein $1$-distance (the norm is called the Kantorovich-Rubinstein norm \cite[Ch 6]{Vill2008}). Explicitly, one first shows that any Wasserstein distance is translation-invariant. Then one shows that the function $f:x\mapsto d^1_W(x+c,c)$ is homogeneous on $\Delta'$. Since every translation-invariant and homogeneous metric is induced by a norm, $d$ induces a norm $N$ on 
$\mathcal{H}$ by setting $N(x)=f(x)$ on $\Delta'$ and then extending it to $\mathcal{H}$ by requiring it to be homogeneous.
\if01
By Proposition~\ref{transinv} it suffices to show that the transportation metric
is homogeneous (in the usual sense of the norm property) over $\Delta'$, and
as we just said, $N$ is its homogeneous extension to $\mathcal{H}$. We
want to show that for any $x\in\Delta'$ and $\lambda$ such that $\lambda
x\in\Delta'$, $d^1_W(\lambda x+c,c)=|\lambda|d^1_W(x+c,c)$. We fix now $\lambda$
and $x$. Let $A$ be a matrix achieving the minimum equal to $d^1_W(x+c,c)$
and satisfying the corresponding conditions. Let $A'$ be the matrix
with marginals $\lambda x+c$ and $c$ such that, for $r\not=r'$,
$A'_{rr'}=\lambda A_{rr'}$ if $\lambda>0$, and
$|\lambda|(^{\operatorname t}\!A)_{rr'}$ otherwise: because of the
conditions on the marginals, we get $\forall r, A'_{rr}=\min(\lambda
x_r+c_r,c_r)$. The matrix $A'$ is nonnegative if and only if $\lambda
x\in\Delta'$: it is nonnegative if and only for all $r$, $\min(\lambda
x_r+c_r,c_r)\geq0$ which occurs if and only if $\lambda x+c$ is in the
simplex. In addition, $\sum_{r,r'\in
S}A'_{r,r'}d_S(r,r')^p=|\lambda|\sum_{r,r'\in S}A_{r,r'}d_S(r,r')^p$, so
we obtain $d^1_W(\lambda x+c,c)\geq|\lambda|d^1_W(x+c,c)$.

On the other hand, $x=\frac{1}{\lambda}{\lambda x}$, and so, by the same
argument, we have that $d^1_W(\lambda x+c,c)\leq|\lambda|d^1_W(x+c,c)$: we have
the equality, and $N$ is a norm, yielding (i).
\fi
 \end{proof}

\begin{remark}
This result is not true for the other Wasserstein metrics, because they
do not satisfy the homogeneity property of norms. For example, let us choose
$z\in\Delta'$ such that $z_1\leq0$ and $\forall i\not=1,
z_{i}\geq0$. It is easy to show that the matrix $A$ reaching the minimum in the definition of $d_W^p(x,y)$ is such that $A_{rr}=min(x_r,y_r)$ for all $r$ and then $d_W^p(x,y)=d_W^p((x-y)^+,(y-x)^+)$ where $x^+=(\max(x_1,0),...\max(x_n,0))$. Then $d_W^p(z+c,c)^p=d_W^p((z_1,0,...0),(0,z_2,...z_n))=\sum_{r'}z_{r'}d(r,r')^p$, and 
$d_W^p(\lambda
z+c,c)=|\lambda|^{\frac{1}{p}}\sum_{r'}z_{r'}d(r,r')^p=
|\lambda|^{\frac{1}{p}}d_W^p(z+c,c)$.
\end{remark}

\begin{prop}
\label{pr:norm2}
The norm $N$ of Proposition~\ref{pr:norm} is not strictly convex.
\end{prop}

\begin{proof} We need to show that the unit sphere is not strictly convex. Fix a ranking $r\in L(C)$ and consider the subset $S_r$ of all
points $x\in \mathcal{H}$ where only the component corresponding to $r$ is negative. In
$S_r$, we have 
$$ N(x) = \sum_{r'} x_{r'} d(r, r'). $$ 
Thus the equation $N(x)=1$ of the unit sphere defines a hyperplane in $S_r$. Since $S_r$ is large, the unit ball is not strictly convex. 
 \end{proof}

We can now show that for any distance $d$, there are two distinct points whose bisector under the Wasserstein distance $d^1_W$ is large.

\begin{prop}
\label{prop:large ell one}
Consider a norm $N$ induced over $\mathcal{H}$ by an $\ell^1$-votewise
metric. Let $r,r_1,r_2$ be rankings. We denote by $d_1$ and $d_2$ the
distances $d(r,r_1)$ and $d(r,r_2)$. Let $\epsilon>0$. We define $x$ and
$y$ as the two points of $\mathcal{H}$ such that
$x_r=-x_{r_1}=\frac{\epsilon}{d_1}$, $y_r=-y_{r_2}=\frac{\epsilon}{d_2}$
and all other components are equal to zero. Then, any point
$z\in\mathcal{H}$ such that $z_r\leq0$ and  $z_{r'}\geq0$ for all $r'\not=r$ is
equidistant from $x$ and $y$ according to $N$. 
\end{prop}
\begin{proof}
Let $z$ be such a point. Then, $x-z$ and $y-z$ have only one positive
component: the one corresponding to the ranking $r$. So
$N(x-z)=\sum_{r'\not=r}(x_{r'}-z_{r'})d(r,r')$ and
$N(y-z)=\sum_{r'\not=r}(y_{r'}-z_{r'})d(r,r')$. Since the only components (different from $r$)
where $x$ and $y$ differ are $r_1$ and $r_2$, they are equidistant from
$z$ if and only if
$(x_{r_1}-z_{r_1})d_1+(x_{r_2}-z_{r_2})d_2=(y_{r_1}-z_{r_1})d_1+(y_{r_2}
-z_{r_2})d_2$, which is equivalent to
$x_{r_1}d_1+x_{r_2}d_2=y_{r_1}d_1+y_{r_2}d_2$, which is in turn
equivalent to $x_{r_1}d_1=y_{r_2}d_2$, which is true by definition. 
 \end{proof}
 
It follows that the behaviour of $\ell^1$-votewise distances is rather counterintuitive.

\begin{cor} \label{cor:large bis}
Let $d$ be an $\ell^1$-votewise metric. Then there is a consensus $\cons$
consisting of isolated points, such that the boundary of $\R(\cons, d)$ is
large. 
\end{cor}
\begin{proof}
Write $x_{\epsilon}$ and $y_{\epsilon}$ for points $x$ and $y$ of the
form defined in Proposition~\ref{prop:large ell one}. That proposition
implies that, if we set $\cons_a=\{x_{\epsilon}\}$ and
$\cons_b=\{y_{\epsilon}\}$ and choose a sufficiently small $\epsilon$,
then $\bis(\cons_a,\cons_b)$ will be large. Also, for any other candidate $c$,
if we set $\cons_c=\{x_{\epsilon_c}\}$ with $\epsilon<\epsilon_c$, then  for
any $z$ such that $z_r$ is the only negative component, $N(z,
\cons_a)=N(z,\cons_b)<N(z,\cons_c)$.
 \end{proof}

\begin{remark}
Note that the consensus in the proof of Corollary~\ref{cor:large bis} is
somewhat unnatural. For example, it is not neutral and does not intersect the boundary of 
$\simp_M$.
\end{remark}

The next question is how often this kind of situation happens. For
simplicity we focus on the case  $d=d_H$, when the induced norm is exactly $\ell^1$. We can give an exact
characterization of when two points have a large bisector. This is
directly connected with the well-known  integer partition problem.

\begin{prop}
\label{bis1}
Let $M\geq 1$ and let $x, y\in\mathbb{R}^M$. We denote by $S$ the
set of values $(x_i-y_i)$. Then $x$ and $y$ have a large bisector under 
$\ell^1$ if and only if there exists a subset $S'\subset S$ such that
$\sum_{e\in S'}e=\sum_{e\not\in S'}e$. 
\end{prop}
\begin{proof}
By definition $\bis(x,y)=\{z|\sum_i|x_i-z_i|=\sum_i|y_i-z_i|\}$. We
divide $\mathbb{R}^M$ into $4^M$ subspaces corresponding to the possible
signs of the values $(x_i-z_i)$ and $(y_i-z_i)$. Let $V$ be one of
these subspaces: in $V$, the equality $\sum_i|x_i-z_i|=\sum_i|y_i-z_i|$
is equivalent to $\sum_i\epsilon_i(x_i-z_i)=\sum_i\epsilon_i'(y_i-z_i)$,
where $\forall i, \epsilon_i,\epsilon'_i=\pm1$. This is 
equivalent to $\sum_i (\epsilon_i - \epsilon'_i) z_i = \sum_{i} (\epsilon_i x_i - \epsilon'_i y_i)$. 

There are two cases. First, if for some $i$, $\epsilon_i\not=\epsilon_i'$, then the linear equation in $z$ is nontrivial and $z$ lies in a hyperplane, so that $V\cap \bis(x,y)$ is small. The other case is when $\epsilon_i=\epsilon_i'$ for all $i$, in which case the left side of the equation is $0$. If the right side is 
nonzero there is no solution, and $V\cap \bis(x,y) = \emptyset$. If the right side is zero, then $V\cap \bis(x,y)$ is large (for each $i$, it contains all points for which $z_i$ is sufficiently large, for example).
The right side is zero if and only if
$\sum_i\epsilon_i(x_i-y_i)=0$, which is equivalent to the fact that
there exists $S'\subset S, \sum_{e\in S'}e=\sum_{e\not\in S'}e$.
\end{proof}

The argument in the proof gives insight into the shape of any large bisector of two points: any ball included in the bisector
is contained in cells where $(x_i-z_i)$ and $(y_i-z_i)$ are of the same sign, and thus in a subset defined by a set of equations
$z_i\leq\min(x_i,y_i)$ or $z_i\geq\max(x_i,y_i)$ for all $i$. It
implies, for example, that if the points are corners of the simplex, 
the large bisector in $\reals^M$ intersects $\simp_M$ in a small set.  Thus, for example, large boundaries cannot occur with
$\sunam$ (which also follows from Corollary~\ref{cor:hyperplane} below).

\begin{defn}
\label{definition:decision}
The standard decision problem PARTITION is defined as follows. Input is a vector $(z_1, \dots, z_M)$ of natural numbers. We must decide whether there is a subset $S\subseteq \{1, \dots, M\}$ for which $\sum_{i\in S} z_i = \sum_{i\not\in S} z_i$.

Define the decision problem LARGE-BISECTOR as follows. Input is a pair
$(x,y)$ of  points of $\rats^M_{+}$ and we must decide whether
$\bis(x,y)$ contains an open ball under the $d^1_H$ metric. 
\end{defn}

\begin{remark}
Note that $M$ is part of
the input in each case. When $M$ is bounded, PARTITION can be solved  trivially by exhaustive enumeration of subsets. Note that if we let $K:=\sum_i x_i$, then a standard dynamic programming algorithm solves PARTITION in $O(KM)$ time.
\end{remark}

\begin{prop}
LARGE-BISECTOR is NP-complete. 
\end{prop}
\begin{proof}
Given an instance $(z_1, \dots, z_M)$ of
PARTITION, let $x_i = z_i, y_i = 0$. This gives an instance of
LARGE-BISECTOR, which is a yes instance if and only if the original instance is a yes instance of PARTITION. Thus LARGE-BISECTOR is NP-hard. On the other hand, given a yes-instance $(x,y,M)$ of
LARGE-BISECTOR, the criterion in  Proposition \ref{bis1} gives a polynomial-sized certificate checkable in polynomial time, so LARGE-BISECTOR is in NP. Thus, LARGE-BISECTOR is NP-complete. 
 \end{proof}

\begin{remark}
We suspect the analogue of LARGE-BISECTOR to be NP-hard for every $\ell^1$-votewise metric. Presumably it is in NP for ``nice" distances, but of course there exist distances which cannot even be computed in polynomial time, so that an analogue of  Proposition \ref{bis1} may not exist.
\end{remark}

The question of large bisectors is quite subtle, because large bisectors do not occur when the consensus sets are hyperplanes instead of points.

\begin{prop} 
\label{prop:bis hyp ell one}
The bisector of two distinct hyperplanes under any norm on $\reals^M$ is contained in a union of at most two hyperplanes.
\end{prop} 
\begin{proof} The distance from a point $x$ to a hyperplane $H$ defined by
$a^Tx=b$ is equal to $d(x,H)=\frac{|a^Tx-b|}{||a||^*}$ where $^*$ denotes the dual norm (see for example
\cite{Mang1997}; the exact definition is not necessary here). Now, let $H'$ be another hyperplane defined by the
equation $a'^Tx=b'$. We assume that $||a||^* = ||a'||^*$
(without loss of generality since multiplying by a scalar still defines
the same hyperplane). The bisector of $H$ and $H'$ can be defined as the
set of points $x$ satisfying $|a^Tx-b| = |a'^Tx-b'|$. 
So, we have two cases, depending on the sign of these absolute values: either
$\sum_i(a_i-a'_i)x_i= b-b'$ or $\sum_i(a_i+a'_i)x_i=b+b'$. Since $H\neq H'$, each of these is the 
equation of a hyperplane.
\end{proof}

\section{Small bisectors and hyperplane rules}
\label{s:small}

All our  results in this section show that the bisectors in question are
contained in a finite union of hyperplanes. Rules which have a
well-defined winner on each component of the complement in $\simp_M$ of a
finite set of hyperplanes have been studied recently. 
Mossel, Procaccia and
Racz \cite{MPR2013} call such simplex rules \hl{hyperplane rules} and show their
equivalence with the \hl{generalized scoring rules} of Xia and Conitzer
\cite{XiCo2008b}. These rules can be defined axiomatically using
\hl{finite local consistency} \cite{XiCo2009}. 
Although originally introduced for social choice rules only, the definition extends to social welfare rules \cite{CPS2014} and it is clear that it also extends to 
general $s$.

Most rules that have ever been studied by social choice theorists are
hyperplane rules. A notable exception is Copeland's rule. In order to
interpret Copeland's rule as a hyperplane rule, Mossel, Procaccia and Racz \cite{MPR2013} require that
the winner be (arbitrarily) specified on the tied region. This seems to
us to be stretching the definition too far -- we could do the same thing for any indecisive rule. 

We now give a sufficient condition for a DR rule to be a hyperplane rule. 

\begin{defn}
\label{def:HMP}
Let $\cons$ be a homogeneous consensus and $d$ a homogeneous distance. 

Say that $(\cons, d)$ satisfies the \hl{votewise minimizer property} (VMP) if the following condition is satisfied.

\if01
\begin{quotation}
For each  $r\in L_s(C)$ and each election $E = (C,V,\pi)\in \elecs$, there
exists a minimizer $(C,V,\pi^*)\in \cons_r$ of the distance from $E$ to
$\cons_r$, so that for all $i, d(\pi_i,\pi^*_i)$ depends only on $\pi_i$
and $r$.
\end{quotation}
\fi

\begin{quotation}
There is a mapping $\xi:L(C^*) \times L_s(C^*) \to L(C^*)$ such that 
for each  $r\in L_s(C^*)$ with $\cons_r\neq\emptyset$, and each election $E = (C,V,\pi)\in \elecs$,
 $m(E, r):=(C,V,\xi(\pi, r))$ minimizes the distance under $d$ from $E$ to
$\cons_r$, where $\xi(\pi, r):=(\xi(\pi_1, r), \dots,  \xi(\pi_n, r))$.
\end{quotation}
\end{defn}

\begin{remark}
The VMP allows us to find a minimizer by dealing with each voter separately, and if two voters have the same preference order in $E$ then the corresponding votes in $m(E, r)$ are equal. Also if $d$ is votewise based on $N$, then
$$d(E, \cons_r) = N(d(\pi_1, \xi(\pi_1, r)), \dots, d(\pi_n, \xi(\pi_n,r)).
$$
If $N$ is also symmetric then $d(E, \cons_r)$ depends only on the multiset of all values $d(\pi_i, \xi(\pi_i, r))$.
\end{remark}

\begin{prop} \label{prop:hyperplane}
Let $d$ be  $\ell^p$-votewise  for some $1\leq p < \infty$ and suppose
that $(\cons, d)$ satisfies the VMP. Then  on $\simp_M$, 
$ \bis(\cons_r, \cons_{r'})$ is defined by
$$
\sum_{t\in L(C)} x_t \delta(t,r)^p =  \sum_{t\in L(C)} x_t\delta(t, r')^p.
$$
\end{prop}
\begin{proof}
The distance between $E = (C,V, \pi)$ and
the minimizer $m(E, r) = (C,V, \pi^*)$ equals $N(\Sigma)$ where $\Sigma$ is the multiset with entries
$\delta(t, r)$ occurring according to their multiplicities $nx_t$,
for all $t \in L(C)$. The specific form of $N$ then shows that 
$d(E, m(E,r))^p = n\left(\sum_t x_t \delta(t, r)^p\right)$. 
Applying the same argument for $r'$ yields the result.
\end{proof}

\begin{defn}
\label{def:gen unam}
Suppose that the $s$-consensus $\cons$ satisfies the following: for each $r\in L_s(C)$, there is a subset $S_r$ of $L(C)$ such that $\cons_r$ consists precisely of those elections for which every voter has a ranking in $S_r$. Then we call $\cons$ a \hl{generalized unanimity} consensus.
\end{defn}

\begin{prop}
\label{prop:VMP}
Let $d$ be an $\ell^p$-votewise distance on $\elecs$ and let $\cons$ be a generalized unanimity consensus. Then $(\cons, d)$ satisfies the VMP.
\end{prop}

\begin{proof} If $\cons_r\neq \emptyset$, define $\xi(\pi_i)$ to be the closest element of $S_r$ to $\pi_i$ under the underlying distance on $L(C)$ (if there is more than one such element, make an  arbitrary choice). For  each $E = (C, V, \pi)$, the element $E^* = (C, V, \xi(\pi))$ belongs to $\cons_r$. If $F=(C, V, \pi')\in \cons_r$ then $d(E, F) \geq d(E, E^*)$ because $d(\pi_i, \xi(\pi_i)) \leq d(\pi_i, \pi'_i)$ for each $i$, and the $\ell^p$-norm is increasing in each argument in the positive orthant. Thus $E^*$ is the desired minimizer.
\end{proof}

\begin{cor}
\label{cor:unam VMP}
$(\sunam^s, d^p)$ satisfies the VMP for every distance $d$.
\end{cor}
\begin{proof}
We can take $S_r$ to be the set of rankings which agree with $r$ in their initial $s$-ranking, showing that $\sunam^s$ is a generalized unanimity consensus.
\end{proof}

\begin{eg}
\label{eg:VMP formula}
Let $\cons = \wunam$ and $d = d_K$, and $N = \ell^2$. For each $E=(C,V,\pi)\in \elecs$ and $a\in C$, we can take $\xi(\pi, a)$ to be the ranking derived from $\pi$ by swapping $a$ to the top as efficiently as possible in each $\pi_i$. 
Thus $d(E, \wunam_a)^2 = \sum_{t\in L(C)} n(t) (\xi(t,a)-1)^2$, where $n(t)$ is the number of times $t$ occurs in $\pi$.
\end{eg}

\begin{remark}
We do not know of any ``natural" consensus and distance which satisfy the VMP, apart from those already mentioned. We can easily create strange examples, however, by creating generalized unanimity consensuses. If for each candidate $a$ we choose a single ranking with $a$ at the top, this yields a generalized unanimity consensus that is extended by $\wunam$. Note that this consensus is not neutral. Alternatively, we could choose all rankings having $a$ in the first or second position (in which case $a$ is the consensus winner), or $b$ in the first position as long as $a$ is not in the second position (in which case $b$ is the consensus winner), or $c$ in the first or second position (provided $a$ is not first or second and $b$ is not first), in which case $c$ is the consensus winner. Again, this is not neutral. 
\end{remark}

\begin{cor}
\label{cor:hyperplane}
Suppose that $d$ is $\ell^p$-votewise with $1\leq p < \infty$, $d$ is finite and not identically zero, and $\cons$ is a generalized unanimity consensus. Then $\R(\cons, d)$ is a hyperplane rule. 
\end{cor}
\begin{proof} Since $d < \infty$ we may rearrange the formula in Proposition~\ref{prop:hyperplane} to get $\sum_t (\delta(t,r) - \delta(t,r')) = 0$. It suffices to show
that the linear function on the left side is not
identically zero. That could only happen if $\delta(t, r) = \delta(t,r')$
for all $t$. However, note that the distance from $x$ to $\cons_r$ is attained at a
point $m(x,r)$ where $m(x,r)_t = x_t$ for all $t\not\in S$, and $d(x,\cons_r) = \sum_{t\in S} x_t \delta(t,r)^p$. If $r\neq r'$ then by
definition $S_r \neq S_{r'}$. Thus taking $t\not\in S\cap S'$, without
loss of generality $\delta(t,r) = 0$ and $\delta(t, r') \neq 0$. 
\end{proof}

\begin{cor}
\label{cor:unam hyper}
Every rule of the form $\R(\sunam^s, d^p)$, where $1\leq p < \infty$ and $d$ is a distance on $L(C)$ that is neither infinite nor identically zero, is a hyperplane rule.
\end{cor}
\begin{remark}
This result does not extend to general distances. For example,
Copeland's rule as we have defined it is not a hyperplane rule, yet it can be defined as $\R(\wunam,
d_{RT})$.  Also note that when $p=\infty$, we do not obtain a hyperplane rule. For example, every point $x\in \simp^{\rats}_M$
for which every coordinate is nonzero is equidistant from all $\sunam_r$, so $\R(\sunam, d^\infty)$ is almost maximally 
indecisive. 
\end{remark}

\begin{remark}
Rules of the type described in Proposition~\ref{prop:hyperplane} are rather special. Since the distance to $\cons_r$ is of the form $\sum_t x_t \delta(t,r)^p$, each can be thought of as a differently weighted version of the rule with $p=1$. 
\end{remark}

\section{Discussion and future work}
\label{s:future}

We now summarize what we have learned about the boundary of a DR simplex rule. 
\begin{itemize}
\item Using a pseudometric that is not a metric can easily lead to a large boundary.
\item Large bisectors can occur even with $\ell^2$, if  consensus sets are not separated.
\item Large bisectors can occur with $\ell^1$-votewise rules, even for consensus sets that are isolated points, and it can be difficult to determine whether they do occur.
\item Even when bisectors are large in the ambient space, using consensus sets on the boundary of the simplex often yields small bisectors on the simplex.
\item Even when bisectors are large on the simplex, neutrality often makes 
the boundary of the rule small.
\end{itemize}

We have seen some desirable properties of consensus sets, such as homogeneity and neutrality.
We argue that convexity (defined in the usual way via restriction from $\reals^M$)  of each $\cons_r$ 
is another essential condition. In the following example, it seems ridiculous that $a$ should win at the 
extra point. 
\begin{eg}
\label{eg:nonconvex}
Consider the case $m=3$, and the consensus formed by extending $\wunam$
so that $a$ is the winner whenever $x_{bca} = x_{cba} = 1/2$ (and
similarly for $b, c$). This consensus is anonymous and homogeneous, 
but $\cons_a, \cons_b, \cons_c$ are not convex. 
\end{eg}

\begin{remark}In the simplex model, convexity (over $\rats$) is equivalent to the notion of \hl{consistency}: if we split the voter set into two parts each of which
elects $r$, the original voter set should elect $r$. It rules out the
above example. Note that $\cond$ and $\sunam^s$ are convex. In fact we do not know of a consensus that has been used in the literature that is not convex.
\end{remark}

Based on the above results, we suggest that the following criteria be required of consensus classes in the simplex (anonymity and homogeneity come for free)
\begin{itemize}
\item neutrality
\item convexity
\item separation
\item intersecting the boundary of the simplex
\end{itemize}
while distances should be required to be metrics. 

Note that the separation requirement rules out $\cond$ as a consensus notion. This may of course be somewhat controversial. It may turn out that neutral rules based on $\cond$ and using metrics always have small boundary (we do not know of a counterexample, but have no proof yet). However, it seems strange to consider a situation arbitrarily close to a complete tie among all rankings (the centre of the simplex) to be an election on which a  ``consensus" can be formed.

We saw above that $\ell^1$ votewise distances can lead to major problems with decisiveness. However 
there are many natural examples of $\ell^1$ votewise distances, as we have seen. We do not know of any ``natural" simplex rule satisfying the above requirements for which the boundary is large. However, not all obvious rules have been thoroughly explored.

Systematic exploration of rules $\R(\cons, d)$, where $\cons$ and $d$ satisfy the recommendations above, may prove fruitful in finding new rules with desirable properties. For example, by the results in this paper and \cite{HaWi2016}, the rules $\R(\sunam^s, d^p)$ where $d$ is a neutral metric on permutations, are anonymous, homogeneous, neutral, continuous, hyperplane rules. There are many neutral (also called right-invariant) metrics on permutations we have not discussed here, such as the $\ell^q$-metrics (the cases $q=1,2,\infty$ being called Spearman's footrule, Spearman's rank correlation and the maximum displacement distance), and the Lee distance \cite{DeDe2009}.
Even the rules $\R(\wunam, d^p)$ and $\R(\sunam, d^p)$  have not been fully explored, to our knowledge. 

Even less understood are rules of the form $\R(\cond, \ell^p)$. For example, when $p=1$, we obtain a homogeneous version of the recently described \emph{Voter Replacement Rule} \cite{EFS2012}. Little is known about the Voter Replacement Rule other than that it is not homogeneous \cite{HaWi2016}. 

Beyond the realm of votewise and $\ell^p$ distances, we have already mentioned more general statistical distances. Finally, rules involving various matrix norms on the tournament matrices have not been well studied.

Distance-based aggregation of preferences is a more general procedure than we have studied here: it could be applied with many different input and output spaces \cite{Zwic2014}. If the input consists of the tournament matrix rather than 
the profile, there is a natural hypercube representation of the input in $\binom{m}{2}$ dimensions. Saari \& Merlin \cite{SaMe2000} showed that the Kemeny rule can be 
described in this way using distance rationalization with respect to the $\ell^1$ norm and $\sunam$. This is the same as using an elementwise norm on the weighted tournament matrix, in our framework. When using profiles as input, the simplex geometry is hard enough to visualize that some authors have used a fixed projection to the \emph{permutahedron} and essentially used $\sunam$ as a consensus. The cases $p=2$  (\emph{mean proximity rules}) \cite{Zwic2008, LaSh2014} and  $p=1$ (\emph{mediancenter rules}) \cite{CDGL+2012} have received attention. These can be interpreted in our framework by changing the distance --- detailed formulae might be interesting.

A question which partially motivated the present work remains unanswered. Does (a homogenization of) Dodgson's rule have a ``small and nice" boundary? What about other Condorcet rules $\R(\cond, d)$ where $d$ is a votewise metric, or even rules based on $d_T$, such as the maximin rule?

\bibliography{voting}{}

\begin{thebibliography}{10}

\bibitem{CaNi1986}
D.E. Campbell and S.I. Nitzan.
\newblock Social compromise and social metrics.
\newblock {\em Social Choice and Welfare}, 3(1):1--16, 1986.

\bibitem{CPS2014}
Ioannis Caragiannis, Ariel~D Procaccia, and Nisarg Shah.
\newblock Modal ranking: A uniquely robust voting rule.
\newblock In {\em Proceedings of the 28th AAAI Conference on Artificial
  Intelligence (AAAI)}, pages 616--622, 2014.

\bibitem{CDGL+2012}
Davide~P. Cervone, Ronghua Dai, Daniel Gnoutcheff, Grant Lanterman, Andrew
  Mackenzie, Ari Morse, Nikhil Srivastava, and William~S. Zwicker.
\newblock Voting with rubber bands, weights, and strings.
\newblock {\em Mathematical Social Sciences}, 64(1):11--27, 2012.

\bibitem{DeDe2009}
Michel~Marie Deza and Elena Deza.
\newblock {\em Encyclopedia of distances}.
\newblock Springer, 2009.

\bibitem{EFS2012}
Edith Elkind, Piotr Faliszewski, and Arkadii Slinko.
\newblock {Rationalizations of Condorcet-consistent rules via distances of
  hamming type}.
\newblock {\em Social Choice and Welfare}, 39(4):891--905, 2012.

\bibitem{EFS2015}
Edith Elkind, Piotr Faliszewski, and Arkadii Slinko.
\newblock Distance rationalization of voting rules.
\newblock {\em Social Choice and Welfare}, 45(2):345--377, 2015.

\bibitem{FaHe2009}
Piotr Faliszewski, Edith Hemaspaandra, and Lane~A Hemaspaandra.
\newblock How hard is bribery in elections?
\newblock {\em Journal of Artificial Intelligence Research}, 35:485--532, 2009.

\bibitem{HaWi2016}
Benjamin Hadjibeyli and Mark~C Wilson.
\newblock Distance rationalization of social rules.
\newblock {\em arXiv preprint arXiv:1610.01902}, 2016.

\bibitem{Horv2000}
{A.G}. Horvath.
\newblock On bisectors in {M}inkowski normed spaces.
\newblock {\em Acta Mathematica Hungarica}, 89:233--246, 2000.

\bibitem{LaSh2014}
S\'{e}bastien Lahaie and Nisarg Shah.
\newblock Neutrality and geometry of mean voting.
\newblock In {\em Proceedings of the fifteenth ACM conference on Economics and
  computation}, pages 333--350. ACM, 2014.

\bibitem{LeNi1985}
Ehud Lerer and Shmuel Nitzan.
\newblock Some general results on the metric rationalization for social
  decision rules.
\newblock {\em Journal of Economic Theory}, 37(1):191--201, 1985.

\bibitem{Mang1997}
O.~L. Mangasarian.
\newblock Arbitrary-norm separating plane.
\newblock {\em Operations Research Letters}, 24:15--23, 1997.

\bibitem{Mann1935}
Heinrich Mann.
\newblock Untersuchungen {\"u}ber wabenzellen bei allgemeiner minkowskischer
  metrik.
\newblock {\em Monatshefte f{\"u}r Mathematik}, 42(1):417--424, 1935.

\bibitem{MeNu2008}
Tommi Meskanen and Hannu Nurmi.
\newblock Closeness counts in social choice.
\newblock In Matthew Braham and Frank Steffen, editors, {\em Power, Freedom,
  and Voting}, pages 289--306. Springer Berlin Heidelberg, 2008.

\bibitem{MPR2013}
Elchanan Mossel, Ariel~D. Procaccia, and Mikl\'{o}s~Z. R\'{a}cz.
\newblock A smooth transition from powerlessness to absolute power.
\newblock {\em Journal of Artificial Intelligence Research}, pages 923--951,
  2013.

\bibitem{Nitz1981}
Shmuel Nitzan.
\newblock Some measures of closeness to unanimity and their implications.
\newblock {\em Theory and Decision}, 13(2):129--138, 1981.

\bibitem{Saar1994}
Donald~G. Saari.
\newblock {\em Geometry of voting}, volume~3 of {\em Studies in Economic
  Theory}.
\newblock Springer-Verlag, Berlin, 1994.

\bibitem{Saar1995}
Donald~G. Saari.
\newblock {\em Basic geometry of voting}.
\newblock Springer-Verlag, Berlin, 1995.

\bibitem{SaMe2000}
Donald~G. Saari and Vincent~R. Merlin.
\newblock Changes that cause changes.
\newblock {\em Soc. Choice Welf.}, 17(4):691--705, 2000.

\bibitem{Vill2008}
C\'{e}dric Villani.
\newblock {\em Optimal transport: old and new}, volume 338.
\newblock Springer Science \& Business Media, 2008.

\bibitem{XiCo2008b}
Lirong Xia and Vincent Conitzer.
\newblock Generalized scoring rules and the frequency of coalitional
  manipulability.
\newblock In Lance Fortnow, John Riedl, and Tuomas Sandholm, editors, {\em
  Proceedings 9th {ACM} Conference on Electronic Commerce (EC-2008), Chicago,
  IL, USA, June 8-12, 2008}, pages 109--118. {ACM}, 2008.

\bibitem{XiCo2009}
Lirong Xia and Vincent Conitzer.
\newblock Finite local consistency characterizes generalized scoring rules.
\newblock In Craig Boutilier, editor, {\em {IJCAI} 2009, Proceedings of the
  21st International Joint Conference on Artificial Intelligence, Pasadena,
  California, USA, July 11-17, 2009}, pages 336--341, 2009.

\bibitem{Youn1975}
H.~Young.
\newblock Social choice scoring functions.
\newblock {\em SIAM Journal on Applied Mathematics}, 28(4):824--838, 1975.

\bibitem{Youn1995}
Peyton Young.
\newblock Optimal voting rules.
\newblock {\em Journal of Economic Perspectives}, 9(1):51--64, 1995.

\bibitem{Zwic2008b}
William~S Zwicker.
\newblock A characterization of the rational mean neat voting rules.
\newblock {\em Mathematical and Computer Modelling}, 48(9):1374--1384, 2008.

\bibitem{Zwic2008}
William~S. Zwicker.
\newblock Consistency without neutrality in voting rules: {When} is a vote an
  average?
\newblock {\em Mathematical and Computer Modelling}, 48(9):1357--1373, 2008.

\end{thebibliography}
\bibliographystyle{plain}
\end{document}